\documentclass{llncs}

\usepackage[ruled,vlined,linesnumbered,commentsnumbered]{algorithm2e}
\usepackage{algcompatible}
\usepackage{amsmath,amssymb,amsfonts}

\input{epsf}

\newcommand {\ignore} [1] {}

\newcommand{\al}    {\alpha}
\newcommand{\be}    {\beta}
\newcommand{\de}    {\delta}

\newcommand{\De}    {\Delta}
\newcommand{\ka}    {\kappa}
\newcommand{\Ga}    {\Gamma}

\newcommand{\subs}  {\subseteq}
\newcommand{\sem}   {\setminus}
\newcommand{\empt} {\emptyset}

\newcommand{\rskc} {\sc Rooted Subset $k$-Connectivity}
\newcommand{\skc} { \sc Subset $k$-Connectivity}

\begin{document}

\title{Improved approximation algorithms for $k$-connected $m$-dominating set problems}

\author{Zeev Nutov} 
\institute{The Open University of Israel. \email{nutov@openu.ac.il.}}

\maketitle

\begin{abstract}
A graph is {\bf $k$-connected} if it has $k$ internally-disjoint paths between every pair of nodes. 
A subset $S$ of nodes in a graph $G$ is a {\bf $k$-connected set} if the subgraph $G[S]$ induced by $S$ is $k$-connected;
$S$ is an {\bf $m$-dominating set} if every $v \in V \sem S$ has at least $m$ neighbors in $S$.
If $S$ is both $k$-connected and $m$-dominating then 
$S$ is a {\bf $k$-connected $m$-dominating set}, or {\bf $(k,m)$-cds} for short.
In the {\sc $k$-Connected $m$-Dominating Set} ({\sc $(k,m)$-CDS}) problem the goal is 
to find a minimum weight $(k,m)$-cds in a node-weighted graph.
We consider the case $m \geq k$ and obtain the following approximation ratios. 
For unit disc-graphs we obtain ratio $O(k\ln k)$, improving the ratio $O(k^2 \ln k)$ of \cite{Fu,ZW}.
For general graphs we obtain the first non-trivial approximation ratio $O(k^2 \ln n)$. 
\end{abstract}

\section{Introduction} \label{s:intro}

A graph is {\bf $k$-connected} if it has $k$ internally disjoint paths between every pair of its nodes. 
A subset $S$ of nodes in a graph $G$ is a {\bf $k$-connected set} if the subgraph $G[S]$ induced by $S$ is $k$-connected;
$S$ is an {\bf $m$-dominating set} if every $v \in V \sem S$ has at least $m$ neighbors in $S$.
If $S$ is both $k$-connected and $m$-dominating set then 
$S$ is a {\bf $k$-connected $m$-dominating set}, or {\bf $(k,m)$-cds} for short.
A graph is a {\bf unit-disk graph}  if its nodes can be located in 
in the Euclidean plane such that there is an edge between nodes $u$ and $v$ 
iff the Euclidean distance between $u$ and $v$ is at most $1$.
We consider the following problem for $m \geq k$ both in general graphs and in unit-disc graphs.

\medskip

\begin{center} \fbox{\begin{minipage}{0.965\textwidth} \noindent
{\sc $k$-Connected $m$-Dominating Set} ({\sc $(k,m)$-CDS}) \\
{\em Input:}  A graph $G=(V,E)$ with node weights $\{w_v:v \in V\}$ and integers $k,m$. \\
{\em Output:}   A minimum weight $(k,m)$-cds $S \subs V$.
\end{minipage}}\end{center}

The case $k=0$ is the {\sc $m$-Dominating Set} problem. % or {\sc $m$-DS} for short.
Let $\al_m$ denote the best known ratio for {\sc $m$-Dominating Set}; 
currently $\al_m=O(1)$ in unit-disc graphs \cite{Fu} and $\al_m=\ln(\De+m) + 1 < \ln \De + 1.7$ in general graphs \cite{Fo},
where $\Delta$ is the maximum degree of the input graph.
The {\sc $(k,m)$-CDS} problem with $m \geq k$ was studied extensively. 
% see \cite{??} for a sample of papers in the area.
In recent papers Zhang, Zhou, Mo, and Du \cite{ZW} and Fukunaga \cite{Fu} obtained 
ratio $O(k^2 \ln k)$ for the problem in unit-disc graphs. 
For unit-disc graphs and $k=2$ Zhang et al. \cite{ZW} also obtained an improved ratio $\al_m+5$.
In a related paper Zhang et al. \cite{ZU} obtained ratio $O(k \ln \Delta)$ in general graphs with unit weights,
mentionning that no non-trivial approximation algorithm for arbitrary weights is known.

Let us say that a graph with a designated set $T$ of terminals and a root node $r$ 
is {\bf $k$-$(T,r)$-connected} if it contains $k$ internally-disjoint $rt$-paths for every $t \in T$. 
Our ratios for {\sc $(k,m)$-CDS} are expressed in terms of $\al_m$ and the best ratio for the following known problem:

\begin{center} \fbox{\begin{minipage}{0.965\textwidth} \noindent
{\rskc}  \\
{\em Input:}  A graph $G=(V,E)$ with edge-costs/node-weights, a set $T \subs V$ of terminals, 
a root node $r \in V \sem T$, and an integer $k$. \\
{\em Output:}   A minimum cost/weight $k$-$(T,r)$-connected subgraph of $G$. 
\end{minipage}}\end{center}

% \medskip

% \begin{center} \fbox{\begin{minipage}{0.965\textwidth} \noindent
% {\rskc}  \\
% {\em Input:}  A graph $G=(V,E)$ with edge-costs/node-weights, $s,t \in V$, and an integer $k$. \\
% {\em Output:}   A minimum cost/weight subgraph of $G$ that contains $k$ internally-disjoint $st$-paths. 
% \end{minipage}}\end{center}

Let $\be_k$ and $\be'_k$ denote the best known ratios for the {\rskc} problem with edge-costs and node-weights, respectively.
Currently, $\be_m=O(1)$ in unit-disc graphs \cite{Fu}, 
while in general graphs $\be_2=2$ \cite{FJW}, $\be_3=6\frac{2}{3}$ \cite{N-CSR},
and $\be_k=O(k\ln k)$ for $k \geq 4$ \cite{N-TALG}.
We also have $\be'_k=O(k^2 \ln n)$ by \cite{N-TALG} and the correction of Vakilian \cite{Vak} to the algorithm
and the analysis of \cite{N-TALG}; see also \cite{Fu-NW}.

Our main results are summarized in the following theorem.

\begin{theorem} \label{t:main}
Suppose that the {\sc $m$-Dominating Set}  problem admits ratio $\al_m$ and that
the {\rskc} problem admits ratios $\be_k$ for edge-costs and $\be'_k$ for node-weights.
Then {\sc $(k,m)$-CDS} with $m \geq k$ admits ratios 
$\al_m+\be'_k+2(k-1)=O(k^2 \ln n)$ for general graphs and 
$\al_m+5\be_k+2(k-1)=O(k\ln k)$ for unit-disc graphs.
Furthermore, {\sc $(3,m)$-CDS} on unit-disc graphs admits ratio $\al_m+5\be_3=\al_m+33\frac{1}{3}$.
\end{theorem}

Our algorithm uses the main ideas as well as partial results from the papers of Zhang et al. \cite{ZW} and Fukunaga \cite{Fu}.
Let us say that a graph $G$ is {\bf $k$-$T$-connected}
if $G$ contains $k$ internally-disjoint paths between every pair of nodes in $T$.
Both papers  \cite{ZW,Fu} consider unit-disc graphs and reduce the {\sc $(k,m)$-CDS} problem with $m \geq k$
to the {\skc} problem:
given a graph with edge costs and a subset $T$ of terminals, find a minimum cost $k$-$T$-connected subgraph.
The problem admits a trivial ratio $|T|^2$ for both edge-costs and node-weights,
while for $|T|>k$ the best known ratios are 
$\frac{|T|}{|T|-k}O(k \ln k)=O(k^2 \ln k)$ for edge-costs and 
$\frac{|T|}{|T|-k}O(k^2 \ln n)=O(k^3 \ln n)$ for node-weights \cite{N-subs}; see also \cite{L-subs}.
In fact, these ratios are derived by applying $O(k)$ times the algorithm for the {\rskc} problem.
The main reason for our improvement over the ratios of \cite{ZW,Fu}
is a reduction to the easier {\rskc} problem.
For small values of $k$ we present a refined reduction,
but for unit disc graphs and $k=2$ the performance of our algorithm and that of \cite{ZW}
coincide, since for $k=2$ and edge-costs both {\skc} and {\rskc} admit ratio $2$ \cite{FJW}.

\section{Proof of Theorem~\ref{t:main}}

% We need some notation and preliminary statements.
For an arbitrary graph $H=(U,F)$ and $u,v \in U$ let $\ka_H(u,v)$ denote the 
maximum number of internally disjoint $uv$-paths in $H$.
We say that $H$ is {\bf $k$-in-connected to $r$} if $H$ is $k$-$(U \sem \{r\},r)$-connected, namely, if
$\ka_H(v,r) \geq k$ every $v \in U \sem \{r\}$. 
For $A \subs U$ let $\Ga_H(A)$ denote the set of neigbors of $A$ in $H$.
The proof of the following known statement can be found in \cite{KN1}, see also \cite{ADNP,DN};
part~(i) of the lemma relies on the Mader’s Undirected Critical Cycle Theorem \cite{mad-cycle}.

\begin{lemma} \label{l:H}
Let $H_r$ be $k$-in-connected to $r$ and let $R=\Ga_{H_r}(r)$.
\begin{itemize}
\item[{\em (i)}]
The graph $H=H_r \sem \{r\}$ can be made $k$-connected by adding a set $J$ of new edges on $R$;
furthermore, if $J$ is inclusionwise-minimal then $J$ is a forest.
\item[{\em (ii)}]
Suppose that $|R|=k$.
If $k=2,3$ then $H_r$ is $k$-connected.
% and if $k=4,5$ then there are $u,v \in R$ such that $H_r \cup \{uv\}$ is $k$-connected.
\end{itemize}
\end{lemma}

Note that an inclusionwise-minimal edge set $J$ as in Lemma~\ref{l:H}(i) 
can be computed in polynomial time, by starting with $J$ being a clique on $R$ and repeatedly removing 
from $J$ an edge $e$ if $H \cup (J \sem e)$ remains $k$-connected.

% The proof of the following known statement can be found in \cite{ZW}.

A reason why the case $m \geq k$ is easier is given in the following lemma.

\begin{lemma} \label{l:kc}
If a graph $H=(V,E)$ has a $k$-dominating set $T$ such that $H$ is $k$-$T$-connected then $H$ is $k$-connected. 
\end{lemma}
\begin{proof}
By a known characterization of $k$-connected graphs,
it is sufficient to show that $|V \sem (A \cup B)| \geq k$ holds 
for any subpartition $A,B$ of $V$ such that $E$ has no edge between $A$ and $B$.
If both $A \cap T,B \cap T$ are non-empty, this is so since $H$ is $k$-$T$-connected.
Otherwise, if say $A \cap T=\empt$, then since $T$ is a $k$-dominating set we have $|\Ga_H(A)| \geq k$,
and the result follows. 
\qed
\end{proof}

Finally, we will need the following known fact, c.f. \cite{N-TALG}.

\begin{lemma}
Given a pair $s,t$ of nodes in a node-weighted graph $G$,
the problem of finding a minimum weight node set $P_{st}$ such that $G[P_{st}]$ has $k$ 
internally- disjoint $st$-paths admits a $2$-approximation algorithm.
\end{lemma}

For arbitrary $k$, we will show that the following algorithm achieves the desired approximation ratio.

\medskip \medskip

\begin{algorithm}[H]
\caption{$(G=(V,E),w,m \geq k)$}  
\label{alg:main}
compute an $\al_m$-approximate $m$-dominating set $T$ \\ 
% compute a   $\be$-approximate $k$-cover $T'$ of $T$ and add $T'$ to $T$ \\
construct a graph $G_r$ by adding to $G$ a new node $r$ connected to a set $R \subs T$ of $k$ nodes
by a set $F_r=\{rv:v \in R\}$ of new edges \\
compute a $\be'_k$-approximate node set $S \subs V \sem T$ such that the subgraph 
$H_r$ of $G_r$ induced by $T \cup S \cup \{r\}$ is $k$-$(T,r)$-connected \\
let $H=H \sem \{r\}=G[T \cup S]$ and let $J$ be a forest of new edges on $R$ as in Lemma~\ref{l:H}(i)
such that the graph $H \cup J$ is $k$-connected \\
for every $uv \in J$ find a $2$-approximate node set $P_{uv}$ such that $G[T \cup S \cup P_{uv}]$ 
has $k$ internally-disjoint $uv$-paths; let $\displaystyle P= \bigcup_{uv \in J} P_{uv}$ \\
return $T \cup S \cup P$
\end{algorithm}

\medskip \medskip

We now prove that the solution computed is feasible.

\begin{lemma} \label{l:feasible}
The computed solution is feasible, namely, at the end of the algorithm $T \cup S \cup P$ is a $(k,m)$-cds.
\end{lemma}
\begin{proof}
Since $T$ is an $m$-dominating set, so is any superset of $T$.
Thus the node set $T \cup S \cup P$ returned by the algorithm is an $m$-dominating set.

It remains to prove that  $T \cup S \cup P$ is a $k$-connected set.
We first prove that the graph $H_r$ computed at step~3 is $k$-in-connected to $r$.
By Menger's Theorem, $\ka_{H_r}(v,r) \geq k$ iff for all $A \subs T \cup S$ with $v \in A$
\begin{equation} \label{e:A}
|\Ga_{H_r \sem R}(A)|+|A \cap R| \geq k \ .
\end{equation}
Let $\empt \neq A \subs T \cup S$. 
If $A \cap T \neq \empt$ then (\ref{e:A}) holds since $H_r$ is $k$-$(T,r)$-connected. 
If $A \cap S \neq \empt$ then $|\Ga_{H_r \sem R}(A)| \geq m \geq k$, since $T$ is an $m$-dominating set and thus 
every node in $A \cap S$ has at least $m$ neighbors in $T$.
In both cases,  (\ref{e:A}) holds, hence $H_r$ is $k$-in-connected to $r$.

The graph $H \cup J$ is $k$-connected, which implies that the graph $G[T \cup S \cup P]$ is $(T \cup S)$-$k$-connected
and thus $T$-$k$-connected. Furthermore, $T$ is a $k$-dominating set, since $m \geq k$. 
Applying Lemma~\ref{l:kc} on the graph $G[T \cup S \cup P]$ we get that this graphs is $k$-connected, as required.
\qed
\end{proof}

\begin{lemma}
Algorithm~\ref{alg:main} has ratio $\al_m+\be'_k+2(k-1)$.
\end{lemma}
\begin{proof}
Let $S^*$ be an optimal solution to {\sc $(k,m)$-CDS}.
Clearly, $w(T) \leq \al_m w(S^*) \leq \beta'_kw(S^*)$.
We claim that $w(S) \leq \be'_k w(S^* \sem T)$.
For this note that $S^* \sem T$ is a feasible solution to the problem considered at step~3 of the algorithm,
while $S$ is a $\be'_k$-approximate solution.
For the same reason, for each $uv \in J$ the set $S^* \sem (T \cup S)$ is a feasible solution 
to the problem considered at step~5, while the set $P_{uv}$ computed is a $2$-approximate solution;
thus $w(P_{uv}) \leq 2w(S^* \sem (T \cup S))$.
Finally, note that $|J| \leq k-1$, and thus $w(P) \leq 2(k-1) w(S^*)$. The lemma follows.
\qed
\end{proof}

This concludes the proof of the case of general $k$ and general graphs.
Let us now consider unit disc graphs. Then we use the following result of \cite{ZW}.

\begin{theorem} [Zhang, Zhou, Mo, and Du \cite{ZW}] \label{t:ZW}
Any $k$-connected unit-disc graph has a $k$-connected spanning subgraph of maximum degree at most 
$5$ if $k=2$, and at most $5k$ if $k \geq 3$.
\end{theorem}

Note that any $k$-connected graph has minimum degree $k$. 
Thus Theorem~\ref{t:ZW} implies that when searching for a $k$-connected subgraph in a unit disc graph, 
one can convert node-weights to edge-costs while invoking in the ratio only a factor of $5/2$ in the case $k=2$  
and $5$ in the case $k \geq 3$. Specifically, given node weights $\{w_v:v \in V\}$ define edge-costs $c_{uv}=w_u+w_v$.
Then for any subgraph $(S,F)$ of $G$ with maximum degree $\De$ and minimum degree $\de$ we have:
$$
\de w(S) \leq c(F) \leq \De w(S)
$$
since $w_v \geq 0$ for all $v \in V$ and since 
$$
c(F)=\sum_{uv \in E} (w_u+w_v)=\sum_{v \in V} d_F(v)w_v \ .
$$
We may use this conversion in some steps of our algorithm, and specifically in step~3,
which concludes the proof of the case of general $k$ and unit-disc graphs.

In the case $k=3$ we use a result of Mader \cite{mad-degk} 
that any edge-minimal $k$-connected graph has at least $\frac{(k-1)n+2}{2k-1}$ nodes of degree $k$.
At step~3 of the algorithm we ``guess'' such a node $r$ and the $3$ edges incident to $r$ in some edge-minimal optimal solution,
remove from $G$ all other edges incident to $r$, and run step~3 while omitting steps 4 and 5. 
By Lemma~\ref{l:H}(ii) the graph $G[S \cup T]$ will be already $3$-connected.

% \bibliographystyle{abbrv}
% \bibliography{cds}

\end{document}